\documentclass[runningheads]{llncs}
\usepackage{microtype}

\usepackage{geometry}%[margin=1.15in]
\usepackage{graphicx}
\usepackage{amssymb}
\usepackage{amsmath}
\usepackage{epstopdf}
\usepackage{wrapfig}
\usepackage{subfig}
\usepackage{hyperref}
\usepackage{enumerate}
\usepackage{array}
\usepackage{delarray}
\usepackage{eufrak}
\usepackage{mathbbol}
\usepackage{tikz}
\usepackage{mathabx}
\usepackage{amsmath}
\usepackage{lscape}
\usepackage{changepage}
\usepackage{multicol}
\usetikzlibrary{fadings,decorations,decorations.pathreplacing,patterns,arrows}
\newcommand{\N}{\mathbb{N}}
\newcommand{\Z}{\mathbb{Z}}

\newcommand{\1}{\mathbb{1}}
\newcommand{\0}{\mathbb{0}}
\renewcommand{\O}{\mathcal{O}}

\newcounter{tmp}
\newcounter{lemmalarge}

\newcounter{mytheorem2}

\newcounter{lemmayn2wave}

\title{Emergence of wave patterns on Kadanoff Sandpiles\thanks{Partially supported by  IXXI (Complex System Institute, Lyon) and ANR projects Subtile, MODMAD, Dynamite and QuasiCool (ANR-12-JS02-011-01).}}

\titlerunning{Emergence of wave patterns on Kadanoff Sandpiles}

\author{K\'evin Perrot\inst{1,}\inst{2} \and \'Eric R\'emila\inst{3}}
\authorrunning{K. Perrot and \'E. R\'emila}

\institute{
Universit\'e de  Lyon - LIP - (umr 5668 - CNRS - ENS de Lyon - Universit\'e Lyon 1) - 46 all\'e d'Italie 69364 Lyon Cedex 7, France
\and
Universit\'e Nice Sophia Antipolis - Laboratoire I3S - UMR 6070 CNRS - 2000 route des Lucioles, BP 121, F-06903 Sophia Antipolis Cedex, France\\
\and
Universit\'e de  Lyon - Groupe d'Analyse de la Th\'eorie Economique  Lyon Saint-Etienne - (umr  5824 - CNRS - Universit\'e Lyon 2) - Site st\'ephanois, 6 rue Basse des Rives, 42 023 Saint-Etienne Cedex 2, France \\
\email{kevin.perrot@ens-lyon.fr \hspace{1cm} eric.remila@univ-st-etienne.fr}
}
\begin{document}

\maketitle

\begin{abstract}
Emergence is a concept that is easy to exhibit, but very hard to formally handle. This paper is about cubic sand grains moving around on nicely packed columns in one dimension (the physical sandpile is two dimensional, but the support of sand columns is one dimensional). The Kadanoff Sandpile Model is a discrete dynamical system describing the evolution of a finite number of stacked grains ---as they would fall from an hourglass--- to a stable configuration (fixed point). Grains move according to the repeated application of a simple local rule until reaching a fixed point. The main interest of the model relies in the difficulty of understanding its behavior, despite the simplicity of the rule. In this paper we prove the emergence of wave patterns periodically repeated on fixed points. Remarkably, those regular patterns do not cover the entire fixed point, but eventually emerge from a seemingly highly disordered segment. The proof technique we set up associated arguments of linear algebra and combinatorics, which interestingly allow to formally state the emergence of regular patterns without requiring a precise understanding of the chaotic initial segment's dynamic.\vspace{.1cm}
\end{abstract}

%%%%%%%%%%%%%%%%%%%%%%%%%%%%%%%%%%
%%
%%	INTRODUCTION
%%
%%%%%%%%%%%%%%%%%%%%%%%%%%%%%%%%%%

\section{Introduction}\label{s:introduction}

Understanding and proving properties on discrete dynamical systems (DDS) is challenging, and demonstrating the global behavior of a DDS defined with local rules is at the heart of our comprehension of natural phenomena \cite{weaver,grauwin}. Sandpile models are a class of DDS defined by local rules describing how grains move in discrete space and time. We start from a finite number of stacked grains ---in analogy with an hourglass\footnote{After reading the definition Section \ref{ss:definition}, see Appendix \ref{a:hourglass} for details.}---, and try to predict the asymptotic shape of stable configurations. 

Bak, Tang and Wiesenfeld introduced sandpile models as systems presenting {\em self-organized criticality} (SOC), a property of dynamical systems having critical points as attractors \cite{bak88}. Informally, they considered the repeated addition of sand grains on a discretized flat surface. Each addition possibly triggers an avalanche, consisting of grains falling from column to column according to simple local rules, and after a while a heap of sand has formed. SOC is related to the fact that a single grain addition on a stabilized sandpile has a hardly predictable consequence on the system, on which fractal structures may emerge \cite{creutz96}. This model can be naturally extended to any number of dimensions.

%%%%%%%%%%%%%%%%%%%%%%%%%%%%%%%%%%
%
\subsection{Kadanoff Sandpile Model (KSPM)}\label{ss:definition}

A one-dimensional sandpile configuration can be represented as a sequence $(h_i)_{i \in \N}$ of non-negative integers, $h_i$ being the number of sand grains stacked on column $i$. The evolution starts from the initial configuration $h$ where $h_0=N$ and $h_i=0$ for $i > 0$, and in the classical sandpile model a grain falls from column $i$ to column $i+1$ if and only if the height difference $h_i - h_{i+1} > 1$. One-dimensional sandpile models were well studied in recent years \cite{goles93,durandlose98,phan04,formenti07,phan08,PSSPM,formenti11}.

\begin{wrapfigure}{r}{4.4cm}
  \centering \includegraphics{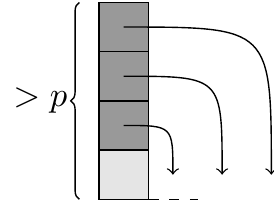}
  \caption{KSPM($p$) rule. When $p$ grains leave column $i$, the slope $b_{i-1}$ is increased by $p$, $b_i$ is decreased by $p+1$ and $b_{i+p}$ is increased by 1. The slope of other columns are not affected.}
  \label{fig:rule}
\end{wrapfigure}

Kadanoff {\em et al.} proposed a generalization of classical models in which a  fixed parameter $p$ denotes the number of grains falling at each step \cite{kadanoff89}. Starting from the initial configuration composed of $N$ stacked grains on column 0, we iterate the following rule: if the difference of height (the slope) between column $i$ and $i+1$ is greater than $p$, then $p$ grains can fall from column $i$, and one grain reaches each of the $p$ columns $i+1,i+2,\dots,i+p$ (Figure \ref{fig:rule}). The rule is applied once (non-deterministically) during each time step.

Formally, this rule is defined on the space of ultimately null decreasing integer sequences where each integer represents a column of stacked sand grains. Let $h=(h_i)_{i \in \mathbb N}$ denote a {\em configuration} of the model, $h_i$ is the number of grains on column $i$. The words {\em column} and {\em index} are synonyms. In order to consider only the relative height between columns, we represent configurations as sequences of {\em slopes} $b=(b_i)_{i \in \N}$, where for all $i \geq 0,~ b_i=h_i-h_{i+1}$. This latter is the main representation we are using (also the one employed in the definition of the model), within the space of ultimately null non-negative integer sequences. We denote by $0^\omega$ the infinite sequence of zeros that is necessary to explicitly write the value of a configuration.

\begin{definition}
  KSPM with parameter $p>0$, KSPM($p$), is defined by two sets:
  \begin{itemize}
    \item \emph{Configurations}. Ultimately null non-negative integer sequences.
    \item \emph{Transition rules}. There is a possible transition from a configuration $b$ to a configuration $b'$ on column $i$, and we note   $b \overset{i}{\rightarrow} b'$ when:
    \begin{multicols}{2}
    \begin{itemize}
      \item $b'_{i-1}=b_{i-1} + p$ (for $i \neq 0$)
      \item $b'_i = b_i - (p+1)$
      \item $b'_{i+p} = b_{i+p} + 1$
      \item $b'_j = b_j$ for $j \not\in  \{i-1, i, i+p \}$. 
    \end{itemize}
    \end{multicols}
  \end{itemize}
\end{definition}
In this case we say that $i$ is {\em fired}. For the sake of imagery, we always consider indices to be increasing on the right (Figure \ref{fig:rule}). Remark that according to the definition of the transition rules, $i$ may be fired if and only if $b_i > p$, otherwise $b'_i$ is negative. We  note $b \rightarrow b'$ when there exists an integer $i$ such that $b \overset{i}{\rightarrow} b'$. The transitive closure of $\rightarrow$ is denoted by  $\overset{*}{\rightarrow}$, and we say that $b'$ is {\em reachable} from $b$ when $b \overset{*}{\to} b'$. A basic property of the KSPM model is the \emph{diamond property} : if there exists $i$ and $j$ such that $b \overset{i}{\rightarrow} b'$ and $b \overset{j}{\rightarrow} b''$, then there exists a configuration $b'''$ such that $b' \overset{j}{\rightarrow} b'''$  and $b'' \overset{i}{\rightarrow} b'''$.

We say that a configuration $b$ is \emph{stable}, or a \emph{fixed point}, if no transition is possible from $b$. As a  consequence of the diamond property, one can easily check that, for each configuration $b$, there exists a unique stable configuration, denoted by $\pi(b)$, such that  $b \overset{*}{\rightarrow} \pi(b)$. Moreover, for any configuration $b'$ such that $b \overset{*}{\rightarrow} b'$, we have $\pi(b') = \pi(b)$ (see \cite{goles02} for details). For convenience, we denote by $N$ the initial configuration $(N,0^\omega)$, such that $\pi(N)$ is the sequence of slopes of the fixed point associated to the initial configuration composed of $N$ stacked grains. This paper is devoted to the study of $\pi(N)$ according to $N$. An example of evolution is pictured on figure \ref{fig:example}.

\begin{figure}[!h]
  \begin{center}
    \includegraphics[width=\textwidth]{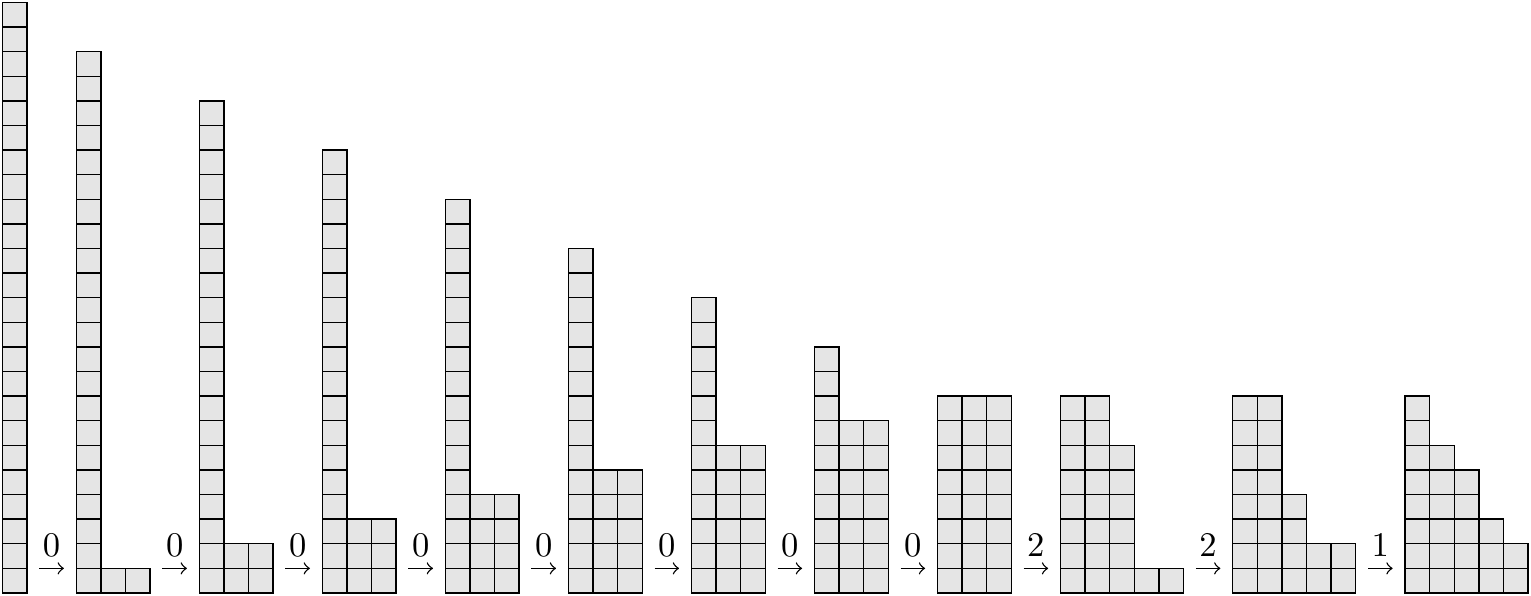}
  \end{center}
  \caption{A possible evolution in KSPM($2$) from the initial configuration for $N\!=\!24$ to $\pi(24)$. $\pi(24)=(2,1,2,1,2,0^\omega)$ and its shot vector (definition in Subsection \ref{ss:dds}) is $(8,1,2,0^\omega)$.}
  \label{fig:example}
\end{figure}

%%%%%%%%%%%%%%%%%%%%%%%%%%%%%%
%
\subsection{Our result}

For a configuration $b$ we denote by $b_{[n,\infty[}$ the infinite subsequence of $b$ starting from index $n$, and $^*$ is the Kleene star denoting finite repetitions of a regular expression (see for example \cite{hopcroft} for details). In this paper we prove the following precise asymptotic form of fixed points, presenting an emergent regular structure stemming from a seemingly complex initial segment (note that the support of $\pi(N)$ is in $\Theta(\sqrt{N})$, see Appendix \ref{a:support} for details):

\setcounter{mytheorem2}{\value{theorem}}
\begin{theorem}\label{theorem:main}
  There exists an $n$ in $\O(\log N)$ such that
  $$\pi(N)_{[n,\infty[} \in (p \cdot \ldots \cdot 2 \cdot 1)^* \,0\, (p \cdot \ldots \cdot 2 \cdot 1)^*\,0^\omega.$$
\end{theorem}

An example for $\pi(2000)$ is given on Figure \ref{fig:2000-h} of Appendix \ref{a:2000}.%This result is fully developed in \cite{kspm-full}.

The result above  presents  an interesting feature: we asymptotically completely describe the form of stable configurations, though there is a part of asymptotically null relative size which remains mysterious. Furthermore, proven regularities are directly stemming from this messy part. Informally, it means that we prove the emergence of  a very regular behavior, after a short transitional and complex phase. Most interestingly, the proof technic we develop does not require to understand precisely this complex initial segment.

In some previous works (\cite{LATA,MFCS}), we obtained a similar result for the smallest parameter $p= 2$ (the case $p=1$ is the well known Sandpile Model) using arguments of combinatorics, but, for the general case, we have to introduce a completely different approach. 
The main ideas are the following: we first relate different representations of a sandpile configuration (Subsection \ref{ss:dds}), which leads to the construction of a DDS on $\Z^{p+1}$ such that the orbit of a well chosen point (according to the number of grains $N$) describes the fixed point configuration we want to characterize. This system is quasi-linear in the sense that the image of a point is obtained by a linear contracting transformation followed by a rounding (in order to remain in $\Z^{p+1}$) which we do not precisely predict. We want to prove that this system converges rapidly, but the unknown rounding makes the analysis of the system very difficult (except for $p = 2$). The key point (Subsection \ref{ss:version}) is the reduction of this system to another quasi-linear system in $\Z^{p}$, for which we have a clear intuition (Subsection \ref{ss:study}), and which allows to conclude the convergence of the system to points involving wavy patterns on fixed points (Subsections \ref{ss:wave} and \ref{ss:proof}).

%%%%%%%%%%%%%%%%%%%%%%%%%%%%%%%
%
\subsection{The context}

 The problem of describing and proving regularity properties suggested by numerical simulations, for models issued from  basic dynamics is a present challenge for physicists, mathematicians, and computer scientists. There exist a lot of conjectures on discrete dynamical systems with simple local rules (sandpile model \cite{dartois} or chip firing games, but also  rotor router  \cite{levine},  the famous Langton's ant \cite{gajardo,propp}...)  but very few results have actually been proved. Regarding KSPM(1), the {\em prediction problem} (namely, the problem of computing the fixed point $\pi(k)$) has been proven in \cite{moore99} to be in \textbf{NC}$^2$ $\subseteq$ \textbf{AC}$^2$ for the one dimensional case, the model of our purpose (improved to \textbf{LOGCFL} $\subseteq$ \textbf{AC}$^1$ in \cite{miltersen}), and \textbf{P}-complete when the dimension is $\geq 3$.  A recent study \cite{goles10} showed that in the two dimensional case the avalanche problem (given a configuration $\sigma$ and two columns $i$ and $j$, does adding one grain on column $i$ have an influence on columnn $j$?) is \textbf{P}-complete for KSPM($p$) with $p>1$, which points out an inherently sequential behavior. The two dimensional case for $p=1$ is still open, though we know from \cite{2006-Goles-CrossingInfo2DSandpile} that wires cannot cross.

%%%%%%%%%%%%%%%%%%%%%%%%%%%%%%%%%%
%%
%%	SNOWBALL
%%
%%%%%%%%%%%%%%%%%%%%%%%%%%%%%%%%%%

\section{Analysis}\label{s:snowball}

We consider the parameter $p$ to be fixed. We study the ``{\em internal dynamic}'' of fixed points, via the construction of a DDS in $\Z^{p+1}$, such that the orbit of a well chosen point according to the number of grains $N$ describes $\pi(N)$. The aim is then to prove the convergence of this orbit in $\O(\log N)$ steps, such that the values it takes involve the form described in Theorem \ref{theorem:main}.

%%%%%%%%%%%%%%%%%%%%%%%%%%%%%%%%%%
%
\subsection{Internal dynamic of fixed points}\label{ss:dds}

%In this subsection we construct a DDS in $\Z^{p+1}$ such that the orbit of a particular point (chosen according to $N$) describes $\pi(N)$.

A useful representation of a configuration reachable from $(N,0^\omega)$ is its {\em shot vector} $(a_i)_{i \in \N}$, where $a_i$ is the number of times that the rule has been applied on column $i$ from the initial configuration (see figure \ref{fig:example} for an example). A fixed point $\pi(N)$ can also be represented as a sequence of slopes $(b_i)_{i \in \N}$ ({\em i.e.}, $b_i=\pi(N)_i$ for all $i$), and those two representations are obviously linked in various ways. In particular for any $i$ we can compute the slope at index $i$ provided the number of firings at $i-p$, $i$ and $i+1$, because $b_i$ is initially equal to 0 (the case $i=0$ is discussed below), and: a firing at $i-p$ increases $b_i$ by 1; a firing at $i$ decreases $b_i$ by $p+1$; a firing at $i+1$ increases $b_i$ by $p$; and any other firing has no consequence on the slope $b_i$. Therefore, $b_i = a_{i-p}  - (p+1) \, a_i + p \, a_{i+1}$, with $ 0 \leq b_i  \leq p$ since $\pi(N)$ is a fixed point, and thus
$$a_{i+1} =   -\frac{1}{p}\, a_{i-p} + \frac {p+1}{p}\, a_i +  \frac {1}{p}\, b_i$$

This equation expresses the value of the shot vector at position $i\!+\!1$ according to its values at positions $i\!-\!p$ and $i$, and a bounded perturbation $0 \leq \frac{b_i}{p} \leq 1$. As an initial condition, we consider a virtual column of index $-p$ that has been fired $N$ times: $a_{-p}=N$ and $a_{i}=0$ for $-p<i<0$, representing the fact that column 0 is the only one receiving $N$ times 1 unit of slope.

\begin{remark}\label{remark:determined}
Note that $a_{i+1} \in \N$, thus $-a_{i-p}+(p+1)\, a_i+b_i \equiv 0 \mod p$. As a consequence, the value of $b_i$ is {\em nearly determined}: given $a_{i-p}$ and $a_i$, there is only one possible value of $b_i$, except when $-a_{i-p} + (p+1)\, a_i \equiv 0 \mod p$ in which case $b_i$ equals $0$ or $p$.
\end{remark}

For example, consider $\pi(2000)$ for $p=4$ (see Appendix \ref{a:2000}). We have $a_8=120$ and $a_4=189$, so $-a_4+5\, a_8=411 \equiv 3 \mod p$. From this knowledge, $b_8$ is determined to be equal to $1$, so that $a_9=-\frac{1}{4}\, a_4 + \frac{5}{4}\, a_8 +  \frac {1}{4}\, b_8=103$ is an integer.

We rewrite this relation as a linear system we can manipulate easily. $a_{i+1}$ is expressed in terms of $a_{i-p}$ and $a_i$, so we choose to construct a sequence of vectors $(X_i)_{i \in \N}$ with $X_i \in \N^{p+1}$ and such that $X_i =\,^t(a_{i-p}, a_{i-p+1}, \dots, a_i)$ where $^tv$ stands for the transpose of $v$. Note that we consider only finite configurations, so there always exists an integer $i_0$ such that $X_i=\0$ for $i_0 \leq i$, with $\0=\,^t(0, \dots, 0)$.

Given $X_i$ and $b_i$ we can compute $X_{i+1}$ with the relation
$$X_{i+1}=A \, X_i + \frac {b_i}{p} J \hspace{1cm}\text{with}\hspace{1cm} A = \begin{pmatrix} 0 & 1 & & 0 & 0\\ & & \ddots & & \\ 0 & 0 & & 1 & 0\\ 0 & 0 & & 0 & 1\\ -\frac{1}{p} & 0 & & 0 & \frac {p+1}{p} \end{pmatrix} \hspace{.5cm} J=\begin{pmatrix}0\\\vdots\\0\\0\\1\end{pmatrix}$$
in the canonical base $B=(e_0,e_1,\dots,e_p)$, with $A$ a square matrix\footnote{As a convention, blank spaces are 0s and dotted spaces are filled by the sequence induced by its endpoints.} of size $(p+1) \times (p+1)$.

This system expresses the shot vector around position $i+1$ (via $X_{i+1}$) in terms of the shot vector around position $i$ (via $X_i$) and the slope at $i$ (via $b_i$). Thus the orbit of the point $X_0=(N,0,\dots,0,a_0)$ in $\N^{p+1}$ describes the shot vector of the fixed point composed of $N$ grains.

Note that it may look odd to study the sequence $(b_i)_{i \in \N}$ using a DDS whose iterations presuppose the knowledge of $(b_i)_{i \in \N}$. It is actually helpful because of the underlined fact that values $b_i$ are {\em nearly determined} (Remark \ref{remark:determined}): in a first phase we will make no assumption on the sequence $(b_i)_{i \in \N}$ (except that $b_i \leq p$ for all $i$) and prove that the system converges exponentially quickly in $N$; and in a second phase we will see that from an $n$ in $\O(\log N)$ such that the system has converged, the sequence $(b_i)_{i \geq n}$ is {\em determined} to have a regular wavy shape.

The system we get is a linear map plus a perturbation induced by the discreteness of values of the slope. Though the perturbation is bounded by a global constant at each step ($b_i \leq p$ for all $i$ since $\pi(N)$ is a fixed point), it seems that the non-linearity prevents classical methods to be powerful enough to decide the convergence of this model.

We denote by $\phi$ the corresponding transformation from $\Z^{p+1}$ to $\Z^{p+1}$, which is composed of two parts: a matrix and a perturbation. Let $R(x)=x^{p-1} + \frac{p-1}{p}\, x^{p-2} + \dots + \frac{2}{p}\, x+ \frac{1}{p}$, the characteristic polynomial of $A$ is $(1-x)^2\, R(x)$. We can first notice that $1$ is a double eigenvalue. A second remark, which helps to get a clear picture of the system, is that all the other eigenvalues are distinct and less than 1 from Lemma \ref{lemma:roots} (using a bound by Enerstr\"om and Kakeya \cite{enestromkakeya}, see Appendix \ref{a:roots}). We will especially use these remarks in Subsection \ref{ss:study}. Therefore there exists a basis such that the matrix of $\phi$ is in Jordan normal form with a Jordan block of size 2. Then, we could project on the $p-1$ other components to get a diagonal matrix for the transformation, hopefully exhibiting an understandably contracting behavior.

We tried to express the transformation $\phi$ in a basis such that its matrix is in Jordan normal form, but we did not manage to handle the effect of the perturbation expressed in such a basis. Therefore, we rather express $\phi$ in a basis such that the matrix and the perturbation act harmoniously. The proof of the Theorem \ref{theorem:main} is done in three steps:
\begin{enumerate}
  \item the construction of a new dynamical system: we first express $\phi$ is a new basis $B'$, and then project along one component (Subsection \ref{ss:version});
  \item the behavior of this new dynamical system is easily tractable, and we will see that it converges exponentially quickly (in $\O(\log N)$ steps) to a uniform vector (Subsection \ref{ss:study});
  \item finally, we prove that as soon as the vector is uniform, then the wavy shape of Theorem \ref{theorem:main} takes place (Subsections \ref{ss:wave} and \ref{ss:proof}).
\end{enumerate}

%%%%%%%%%%%%%%%%%%%%%%%%%%%%%%%%%%
% 
\subsection{Making the matrix and the perturbation act harmoniously}\label{ss:version}

From the dynamical system $X_{i+1}=A \, X_i +\frac {b_i}{p} J$ in the canonical basis $B$, we construct a new dynamical system for $\phi$ in two steps: first we change the basis of $\Z^{p+1}$ in which we express $\phi$, from the canonical one $B$, to a well chosen $B'$; then we project the transformation along the first component of $B'$. The resulting system on $\Z^p$, called {\em averaging system}, is very easily understandable, very intuitive, and the proof of its convergence to a uniform vector can then be completed straightforwardly.

$$\text{Let }
B'=
\begin{pmatrix}
1 & 0 & & 0\\
1 & 1 & & 0\\
\vdots & \vdots & \ddots & \\
1 & 1 & \dots & 1
\end{pmatrix}
\text{ and }
B'^{-1}=
\begin{pmatrix}
  1 & & 0& 0\\
  -1 & \ddots & 0 & 0 \\
   & \ddots & \ddots & \\
   0 & & -1 & 1\\
\end{pmatrix}
\text{ be square matrices of size } p+1.
$$
$B'=(e'_0,\dots,e'_p)$ (with $e'_i$ the $(i+1)^{th}$ column of the matrix $B'$) is a basis of $\Z^{p+1}$, and we have
$$
\begin{array}[t]{rrcl}
& B'^{-1} \, X_{i+1} & = & B'^{-1} \, A \, B' \, B'^{-1} \, X_i + \frac{b_i}{p} B'^{-1} \, J\\
\iff &X'_{i+1} & = & A' \, X'_i + \frac{b_i}{p} J'
\end{array}
$$
with
$$
X'_i = B'^{-1} \, X_i
\hspace{1cm}
A' = B'^{-1} \, A \, B' =
\begin{pmatrix}
  1 & 1 & & 0\\
   & & \ddots & \\
   0 & 0 & & 1\\
   0 & \frac{1}{p} & \dots & \frac{1}{p}
\end{pmatrix}
\hspace{1cm}
J' = B'^{-1} \, J =
\begin{pmatrix}
0\\ \vdots \\ 0\\ 1
\end{pmatrix}$$

We now proceed to the second step by projecting along $e'_0$. Let $P$ denote the projection in $\Z^{p+1}$ along $e'_0$ onto $\{0\} \times \Z^p$. We can notice that $e'_0$ is an eigenvector of $A'$, hence projecting along $e'_0$ simply corresponds to erasing the first coordinate of $X'_i$. For convenience, we do not write the zero component of objects in $\{0\} \times \Z^p$.

The new DDS we now have to study, which we call {\em averaging system}, is
\begin{eqnarray}
  Y_{i+1} = M \, Y_i + \frac{b_i}{p} K \label{eq:averaging}
\end{eqnarray}
with the following elements in $\Z^p$ (in $\{0\} \times \Z^p$):
$$
Y_i = P \, X'_i
\hspace{1cm}
M=P \, A'=
\begin{pmatrix}
   0 & 1 & & 0\\
   & & \ddots & \\
   0 & 0 & & 1\\
   \frac{1}{p} & \frac{1}{p} & \dots & \frac{1}{p}
\end{pmatrix}
\hspace{1cm}
K=P \, J'=
\begin{pmatrix}
0\\ \vdots \\ 0\\ 1
\end{pmatrix}
$$

Let us look in more details at $Y_i$ and what it represents concerning the shot vector. We have $X_i=\,^t(a_{i-p},a_{i-p+1},\dots,a_i)$, thus
$$Y_i=P \, B'^{-1} \, X_i=
\begin{pmatrix}
a_{i-p+1}-a_{i-p}\\
\vdots\\
a_{i-1}-a_{i-2}\\
a_i-a_{i-1}
\end{pmatrix}
\text{ and for initialization }
Y_0=
\begin{pmatrix}
-N\\
0\\
\vdots\\
0\\
a_0
\end{pmatrix}.$$
$Y_i$ represents differences of the shot vector, which may of course be negative. In Subsection \ref{ss:study} we will see that the averaging system is easily tractable and converges exponentially to a uniform vector. Subsection \ref{ss:wave} concentrates on the implications following this uniform vector, {\em i.e.}, the emergence of a wavy shape.

%%%%%%%%%%%%%%%%%%%%%%%%%%%%%%%%%%
% 
\subsection{Convergence of the averaging system}\label{ss:study}

The averaging system is understandable in simple terms. From $Y_i$ in $\Z^p$, we obtain $Y_{i+1}$ by:
\begin{enumerate}
  \item shifting all the values one row upward;
  \item for the bottom component, computing the mean of values of $Y_i$, and adding a small perturbation (a multiple of $\frac{1}{p}$ between 0 and 1) to it.
\end{enumerate}

Let $y_i$ be the first component of $Y_i$, we therefore have $Y_i=\,^t(y_i,\dots,y_{i+p-1})$.

\begin{remark}\label{remark:determined2}
$(Y_i)_{i \in \N}$ are still integer vectors, hence the perturbation added to the last component is again nearly determined: let $m_i$ denote the mean of value of $Y_i$, we have $(m_i + \frac{b_i}{p}) \in \Z$ and $0 \leq \frac{b_i}{p} \leq 1$. Consequently, if $m_i$ is not an integer then $b_i$ is determined and equals $p(\lceil m_i \rceil - m_i)$, otherwise $b_i$ equals 0 or $p$.
\end{remark}

For example, consider $\pi(2000)$ for $p=4$ (see Figure \ref{fig:2000-dsv} of Appendix \ref{a:2000}, be careful that it pictures $a_i-a_{i+1}$ at position $i$). We have $Y_{13}=\,^t(-3,-5,-7,-7)$, then $y_{13}=-\frac{11}{2}$ and $b_{13}$ is forced to be equal to 2 so that $Y_{14}=\,^t(-5,-7,-7,-5)$ is an integer vector.

We can foresee what happens as we iterate this dynamical system and new values are computed: on a {\em large scale} ---when values are large compared to $p$--- the system evolves roughly toward the mean of values of the initial vector $Y_0$, and on a {\em small scale} ---when values are small compared to $p$--- the perturbation lets the vector wander a little around. Previous developments where intending to allow a simple argument to prove that those wanderings do not prevent the exponential convergence towards a uniform vector.

The study of the convergence of the averaging system works in three steps:
\begin{enumerate}
\item[(i).] state a linear convergence of the whole system; then express $Y_n$ in terms of $Y_0$ and $(b_i)_{0 \leq i \leq n}$;
\item[(ii).] isolate the perturbations induced by $(b_i)_{0 \leq i \leq n}$ and bound them;
\item[(iii).] prove that the other part (corresponding to the linear map $M$) converges exponentially quickly.
\end{enumerate}
From (ii) and (iii), a point converges exponentially quickly into a ball of constant radius, then from (i) this point needs a constant number of extra iterations in order to reach the center of the ball, that is, a uniform vector.

\begin{proposition}\label{lemma:average}
  There exists an $n$ in $\O(\log N)$ such that $Y_n$ is a uniform vector.
\end{proposition}

\begin{proof}
Let $m_i$ (respectively $\overline{m}_i$, $\underline{m}_i$) denote the mean (respectively maximal, minimal) of values of $Y_i$. We will prove that $\overline{m}_i-\underline{m}_i$ converges exponentially quickly to 0, which proves the result.

We start with $Y_0=\,^t(-N,0,\dots,0,a_0)$, thus $\overline{m}_0-\underline{m}_0 = N+a_0 \leq \frac{p+1}{p} N$ since $a_0 \leq \frac{N}{p}$ (recall that $a_0$ is the number of times column 0 has been fired).

This proof is composed of two parts. Firstly, the system converges exponentially quickly on a large scale. Intuitively, when $\overline{m}_i-\underline{m}_i$ is large compared to $p$, the perturbation is negligible.

\begin{adjustwidth}{.5cm}{0cm}
\setcounter{lemmalarge}{\value{lemma}}
\begin{lemma}\label{lemma:large}
  There exists a constant $\alpha$ and $n_0$ in $\O(\log N)$ s.t. $\overline{m}_{n_0} - \underline{m}_{n_0} < \alpha$.
\end{lemma}
\begin{proof}[Proof sketch]
%Complete proof in Appendix \ref{a:convergence}. Let $M_n=(m_n,\dots,m_n)$ in $\Z^p$. Since $Y_n$ converges roughly towards the mean of its values, we consider the evolution of $Z_n=Y_n-M_n$. We therefore establish a relation of the form $Z_{n+1}=N \, Z_n + C_n$ with $C_n$ a bounded perturbation and $N$ diagonalizable with the moduli of all its eigenvalues less than 1 (using a bound by Enestr\"om and Kakeya). We can then upper bound $\sum_{i=0}^{n-1} N^{n-1-i}\,C_i$ by a constant independent of $n$ and the number of grains $N$, and find $n_0$ in $\O(\log N)$ such that $Z_{n_0}=N^{n_0} Z_0 + \sum_{i=0}^{n_0-1} N^{n_0-1-i}\,C_i$ is upper bounded.
%
Complete proof in Appendix \ref{a:convergence}. Let $M_i=(m_i,\dots,m_i)$ in $\Z^p$. Since $Y_i$ converges roughly towards the mean of its values, we consider the evolution of $Z_i=Y_i-M_i$. We easily  establish a relation of the form $Z_{i+1}=O \, Z_i + C_i$ with $C_i$ a bounded perturbation vector and $O$ a linear transformation. It follows that 
 $Z_{n}=O^{n} Z_0 + \sum_{i=0}^{n-1} O^{n-1-i}\,C_i$. 
 
Moreover, the characteristic polynomial of $O$ is $R(x)$ (see Lemma \ref{lemma:roots} in Appendix \ref{a:roots}). One proves that $R(x)$ has $p-1$ distinct roots $\lambda_1,\dots,\lambda_{p-1}$ (using  coprimality of $R(x)$ and $R'(x)$) and for all $k$, $| \lambda_k | \leq \frac{p-1}{p} < 1$ (using a bound by Enestr\"om and Kakeya, see for example \cite{enestromkakeya}).

%the roots of $R(x)$   are pairwise distinct (using  coprimality of $R(x)$ and $R'(x)$),  with  moduli strictly lower  than 1 (using a bound by Enestr\"om and Kakeya (see for example\cite{enestromkakeya}))

 Consequently, $O$ is a contraction operator, and $\sum_{i=0}^{n-1} O^{n-1-i}\,C_i$ can be upper bounded by a constant independent of $n$ and the number of grains $N$. We can therefore conclude that there exists an $n_0$ in $\O(\log N)$ such that $Z_{n_0}=O^{n_0} Z_0 + \sum_{i=0}^{n_0-1} O^{n_0-1-i}\,C_i$ is upper bounded.
\end{proof}
\end{adjustwidth}

Secondly, on a small scale, the system converges linearly.

\begin{adjustwidth}{.5cm}{0cm}
\begin{lemma}\label{lemma:small}
The value of $\overline{m}_i - \underline{m}_i$ decreases linearly: if $\underline{m}_i \neq \overline{m}_i$, then there is an integer $c$, with $0 \leq c \leq p$ such that $\overline{m}_{i+c} - \underline{m}_{i+c} < \overline{m}_i - \underline{m}_i$.
\end{lemma}
\begin{proof}
  If $\underline{m}_i \neq \overline{m}_i$, that is, if the vector $Y_i$ is not uniform, the mean value is strictly between the greatest and smallest values: $\underline{m}_i < m_i < \overline{m}_i$. Consequently $\underline{m}_i < y_{i+p} = m_i + \frac{b_i}{p} \leq \overline{m}_i$ (since the perturbation added is at most one and the resulting number is an integer, we cannot reach a greater integer). Therefore, we get  $\underline{m}_i \leq \underline{m}_{i+1} \leq \overline{m}_{i+1} \leq \overline{m}_{i}$. 
  
  This reasoning applies while $\underline{m}_{i+j} \neq \overline{m}_{i+j}$,  
  from which we get $\underline{m}_{i+j}<  y_{i+p+j} \leq \overline{m}_{i+j}$ and $\underline{m}_{i+j} \leq \underline{m}_{i+j+1} \leq \overline{m}_{i+j+1} \leq \overline{m}_{i+j}$. 
  
If there exists $c \leq p$ such that  $\underline{m}_{i+c}=\overline{m}_{i+c}$,  
  then,  we are done. 
  Otherwise, for  $0 \leq j < p$,  we have, $\underline{m}_i \leq \underline{m}_{i+j} < y_{i + j + p} \leq \overline{m}_{i+j} \leq \overline{m}_i$), thus $\underline{m}_i < \underline{m}_{i+p}  \leq \overline{m}_{i+p} \leq \overline{m}_i$.
\end{proof}
\end{adjustwidth}

To conclude, we start with $\overline{m}_0-\underline{m}_0$ in $\O(N)$, we have a constant $\alpha$ and a $n_0$ in $\O(\log N)$ such that $\overline{m}_{n_0}-\underline{m}_{n_0} < \alpha$ thanks to the exponential decrease on a large scale (Lemma \ref{lemma:large}). Then after $p$ iterations the value of $\overline{m}_{n_0+p}-\underline{m}_{n_0+p}$ is decreased by at least 1 (Lemma \ref{lemma:small}), hence there exists $\beta$ with $\beta \leq p\, \alpha$ such that after $\beta$ extra iterations we have $\overline{m}_{n_0+\beta}-\underline{m}_{n_0+\beta}=0$. Thus $Y_{n_0+\beta}$ is a uniform vector, and $n_0+\beta$ is in $\O(\log N)$.
\end{proof}

In this proof, neither the discrete nor the continuous studies is conclusive by itself. On one hand, the discrete study gives a linear convergence but not an exponential convergence. On the other hand, the continuous study gives an exponential convergence towards a uniform vector, but in itself the continuous part never reaches a uniform vector but tends asymptotically towards it. It is the simultaneous study of those modalities (discrete and continuous) that allows to reach the conclusion.

\begin{remark}\label{remark:p2}
  Note that for $p=1$, the averaging system has a trivial dynamics. For $p=2$, the behavior is a bit more complex, but major simplifications are found: the computed value is equal to the mean of two values, hence in this case the difference $\overline{m}_i-\underline{m}_i$ decreases by a factor of two {\em at each time step}.
\end{remark}

%%%%%%%%%%%%%%%%%%%%%%%%%%%%%%%%%%
% 
\subsection{Emergence of a loosely wavy shape}\label{ss:wave}

We call {\em wave} the pattern $p \cdot \ldots \cdot 2\cdot 1$ in the sequence of slopes. Lemma \ref{lemma:average} shows that there exists an $n=\O(\log N)$ such that $Y_n$ is a uniform vector. In this subsection, we prove that if $Y_n$ is a uniform vector, then the shape of the sandpile configuration is exclusively composed of waves and 0s from the index $n$.

\setcounter{lemmayn2wave}{\value{lemma}}
\begin{lemma}\label{lemma:yn2wave}
  $Y_n$ is a uniform vector of $\Z^p$ implies
  $$\pi(N)_{[n,\infty[} \in \Big(0 + ( p \cdot p\!-\!1 \cdot \ldots \cdot 1 ) \Big)^* 0^\omega$$
\end{lemma}

\begin{proof}[Proof sketch] Complete proof in Appendix \ref{a:yn2wave}.
We straightforwardly apply Remark \ref{remark:determined2}. If $Y_n$ is a uniform vector, we notice that the value of $b_n$ is 0 or $p$. If it is $0$ then $Y_{n+1}$ is still a uniform vector; if it is $p$, then the sequence $(b_i)_{n \leq i < n+p}$ is determined to be equal to $p \cdot p-1 \cdot \ldots \cdot 1$, and $Y_{n+p}$ is a uniform vector. The following diagram illustrates those observations: the grey node represents a uniform vector, and arrows are labelled by values of the sequence $b_i$. If we start on the grey node, any path's labels verify the statement of the Lemma.
\begin{center}
  \includegraphics{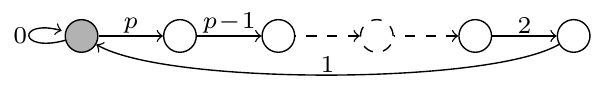}
\end{center}
\end{proof}

\begin{remark}\label{remark:wave}
Composing Proposition \ref{lemma:average} and Lemma \ref{lemma:yn2wave} allows us to prove the emergence, from a logarithmic column, of the shape given in Lemma \ref{lemma:yn2wave}.
\end{remark}

%%%%%%%%%%%%%%%%%%%%%%%%%%%%%%%%%%
% 
\subsection{Avalanches to complete the proof}\label{ss:proof}

In order to prove Theorem \ref{theorem:main}, we refine Remark \ref{remark:wave} to show that there is at most one set of two non empty and consecutive columns of equal height, called {\em plateau} of size two, and corresponding to a slope equal to 0. It seems necessary to overcome the ``static'' study ---for a given fixed point--- presented above, and consider the dynamic of sand grains from $\pi(0)$ to $\pi(N)$. As presented in Appendix \ref{a:hourglass}, $\pi(N)$ can be computed inductively, using the relation
$$\text{for all } k>0,~ \pi(\pi(k-1)^{\downarrow 0})=\pi(k)$$
where $\sigma^{\downarrow 0}$ denotes the configuration obtained by adding one grain on column 0 of $\sigma$. We start from $\pi(0)$ and inductively compute $\pi(1),\pi(2),\dots,\pi(N\!-\!1)$ and $\pi(N)$ by repeating the addition of one grain on column 0 and reaching a stable configuration. The sequence of firings from $\pi(k\!-\!1)^{\downarrow 0}$ to $\pi(k)$ is called the $k^{th}$ {\em avalanche} (see Appendix \ref{a:hourglass} for details).

We studied the structure of avalanches in \cite{LATA,MFCS}, and first proved that each column is fired at most once in an avalanche $s^k$. Secondly, we showed that as soon as $p$ consecutive columns are fired, then the avalanche fires a set of consecutive columns ---without any {\em hole} $i$ such that $i \notin s^k$ and $(i+1) \in s^k$---, and saw that this property leads to important regularities in successive fixed points. The following proof uses those observations: the structure of an avalanches on a wave pattern is very constrained, and as soon as an avalanche goes beyond a wave, it necessarily fires every column of that wave, thus it fires a set of $p$ consecutive columns (formally stated in Corollary \ref{corollary:snowball} of Appendix \ref{a:gdl}). We detail in the following proof why this property of avalanches on wave patterns ensures that if there is at most one plateau of size two (one slope equal to 0) in-between wave patterns of a fixed point $\pi(N-1)$, then there remains at most one plateau of size two on the wave patterns of the fixed point $\pi(N)$.

\begin{proof}[Proof of Theorem \ref{theorem:main}]
  We prove the result by induction on $N$. From Proposition \ref{lemma:average} and Lemma \ref{lemma:yn2wave}, there is an index $n$ (resp. $n'$) in $\O(\log N)$ from which $\pi(N)$ (resp. $\pi(N-1)$) is described by the expression given in Lemma \ref{lemma:yn2wave}. Moreover, from Corollary \ref{corollary:snowball} (see Appendix \ref{a:gdl}), there is an index $l$ in $\O(\log N)$ such that the $N^{th}$ avalanche fires a set of consecutive columns on the right of $l$. Without loss of generality, we consider that $l \geq n,n'$, and will prove that if $\pi(N-1)_{[l,\infty[}$ has at most one value 0, then so has $\pi(N)_{[l,\infty[}$.
  
  Now, if the avalanche ends before column $l-p$ (if $\max s^N < l-p$), the result holds. In the other case, we simply notice by contradiction that as soon as the avalanche reaches the wave patterns, it necessarily ends on the first value 0 it encounters, otherwise the resulting configuration is not stable. The consequence is that the 0 ``climbs'' one wave to the left, preserving the invariant of having at most one value 0 in-between wave patterns.
\end{proof}

%%%%%%%%%%%%%%%%%%%%%%%%%%%%%%%%%%
%%
%%	CONCLUDING DISCUSSION
%%
%%%%%%%%%%%%%%%%%%%%%%%%%%%%%%%%%%

\section{Concluding discussion}\label{s:conclusion}

The proof technic we set up in this paper allowed us to prove the emergence of regular patterns periodically repeated on fixed points, without requiring a precise understanding of the initial segment's dynamic. Arguments of linear algebra allowed to prove a rough convergence of the system (when the dynamic is not precisely known but coarsely bounded), completed with arguments of combinatorics, using the discreteness of the model, to prove the emergence of precise and regular wave patterns.

This result stresses the fact that sandpile models are on the edge between discrete and continuous systems. Indeed, when there are very few sand grains, each one seems to contribute greatly to the global shape of the configuration. However, when the number of grains is very large, a particular sand grain seems to have no importance to describe the shape of a configuration. The result also suggests a separation of the discrete and continuous parts of the system. On one hand, the seemingly unordered initial segment, interpreted as reflecting the discrete behavior, prevents regularities from emerging. On the other hand, the asymptotic and ordered part, interpreted as reflecting the continuous behavior, lets a regular and smooth pattern come into view.

Nevertheless, the separation between discrete and continuous behaviors may be challenged because the continuous part emerges from the discrete part. We have two remarks about this latter fact. Firstly, the consequence seems to be a slight bias appearing on the continuous part: it is not fully homogeneous ---that is, with exactly the same slope at each index--- which would have been expected for a continuous system, but a ---very small--- pattern is repeated. It looks like this bias comes from the gap between the unicity of the {\em border} column on the left side at index $-1$ compared to the rule which has a parameter $p$, because we still observe the appearance of wave patterns starting from variations of the initial configuration (for example starting from $p$ consecutive columns of height $N$, thus $p\,N$ grains). Secondly, if we consider the asymptotic form of a fixed point, the relative size of the discrete part is null. This, regarding the intuition described above that when the number of grains is very large then a particular grain has no importance, is satisfying.

Finally, the emergence of regularities in this system hints at a clear qualitative distinction between some sand grains and a heap of sand. Let us save the last words to a distracting application to the famous {\em sorites paradox}. Someone who has a very little amount of money is called {\em poor}. Someone {\em poor} who receives one cent remains {\em poor}. Nonetheless, if the increase by 1 cent is repeated a great number of times then the person becomes {\em rich}. The question is: when exactly does the person becomes {\em rich}? An answer may be that {\em richness} appears when money creates waves...

\bibliographystyle{plain}
\bibliography{biblio}

\newpage

%%%%%%%%%%%%%%%%%%%%%%%%%%%%%%%%%%
%%
%%	APPENDIX
%%
%%%%%%%%%%%%%%%%%%%%%%%%%%%%%%%%%%

\appendix

\chapter*{\appendixname}

%%%%%%%%%%%%%%%%%%%%%%%%%%%%%%%%%%
%
%	EXAMPLE
%
%%%%%%%%%%%%%%%%%%%%%%%%%%%%%%%%%%

%\section{An example of evolution}\label{a:example}
%
%\begin{figure}[!h]
%  \begin{center}
%    \includegraphics[width=\textwidth]{fig-example.pdf}
%  \end{center}
%  \caption{A possible evolution in KSPM($2$) from the initial configuration for $N\!=\!24$ to $\pi(24)$. $\pi(24)=(2,1,2,1,2,0^\omega)$ and its shot vector is $(8,1,2,0^\omega)$.}
%  \label{fig:example}
%\end{figure}

%%%%%%%%%%%%%%%%%%%%%%%%%%%%%%%%%%
%
%	HOURGLASS
%
%%%%%%%%%%%%%%%%%%%%%%%%%%%%%%%%%%

\section{Hourglass}\label{a:hourglass}

In order to compute $\pi(N)$, the basic procedure is to start from the initial configuration $(N,0^\omega)$ and perform all the possible transitions. However, it also possible to start from the configuration $(0^\omega)$, add one grain on column 0 and perform all the possible transitions, leading to $\pi(1)$, then add another grain on column 0 and perform all the possible transitions, leading to $\pi(2)$, etc... And repeat this process until reaching $\pi(N)$.

Formally, let $b$ be a configuration, $b^{\downarrow 0}$ denotes the configuration obtained by adding one grain on column 0. In other words, if $b=(b_0,b_1,\dots)$ then $b^{\downarrow 0}=(b_0 +1 ,b_1,\dots)$. The correctness of the process described above relies on the fact that
$$(k,0^\omega) \overset{*}{\to} \pi(k-1)^{\downarrow 0}$$
Indeed, there exists a sequence of firings, named a {\em strategy}, $(s_i)_{i=1}^{i=l}$ such that $(k-1,0^\omega) \overset{s_1}{\to} \dots \overset{s_l}{\to} \pi(k-1)$. It is obvious that using the same strategy we have $(k,0^{\omega}) = (k-1,0^\omega)^{\downarrow 0} \overset{s_1}{\to} \dots \overset{s_l}{\to} \pi(k-1)^{\downarrow 0}$ since we only drag one more grain on column 0 along the evolution, which does not prevent any firing (see Figure \ref{fig:inductive}). Thus, with the uniqueness of the fixed point reachable from $(k,0^\omega)$, we have the recurrence formula:
$$\pi(\pi(k-1)^{\downarrow 0})  =   \pi (k)$$
with the initial condition $\pi (0) = 0^\omega$, enabling an inductive computation of $\pi(k)$.

\begin{figure}[!h]
  \centering \includegraphics{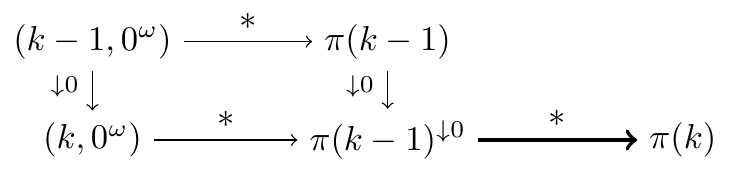}
  \caption{inductive computation of $\pi(k)$ from $\pi(k-1)$. The bold arrow represents an avalanche.}
  \label{fig:inductive}
\end{figure}

The strategy from $\pi(k-1)^{\downarrow 0}$ to $\pi(k)$ is called an {\em avalanche}. Note that due to the non-determinacy of the model, this strategy is not unique. To overcome this issue, it is natural to distinguish a particular one which we think is the simplest: the {\em $k^{th}$ avalanche} $s^k$ is the leftmost strategy from $\pi(k-1)^{\downarrow 0}$ to $\pi(k)$, where {\em leftmost} is the minimal strategy according to the lexicographic order. This means that at each step, the leftmost possible firing is performed. A preliminary result of \cite{LATA} is that any column is fired at most once during an avalanche, which allows to write without ambiguity for an index $i$: $i \in s^k$ or $i \notin s^k$.

%%%%%%%%%%%%%%%%%%%%%%%%%%%%%%%%%%
%
%	THERE IS NO PLATEAU OF LENGTH LARGER THAN p+1
%
%%%%%%%%%%%%%%%%%%%%%%%%%%%%%%%%%%

\section{There is no plateau of length larger than $p+1$}\label{a:plateau}

A plateau is a set of at least two non empty and consecutive columns of equal height. The length of a plateau is the number of columns composing it.

\begin{lemma}\label{lemma:plateau}
  For any $N$ and any configuration $\sigma$ such that $(N,0^\omega) \overset{*}{\to} \sigma$, in $\sigma$ there is no plateau of length strictly greater than $p+1$.
\end{lemma}

\begin{proof}
  This proof proceeds by contradiction, using the fact that configurations are sequences of non-negative integers ($\mathcal H_1$). Suppose there exists a configuration $\sigma$ reachable from $(N,0^\omega)$ for some $N$, such that there is a plateau of length at least $p+2$ in $\sigma$. Since there is no plateau in the initial configuration $(N,0^\omega)$, and there is a finite number of steps to reach $\sigma$, there exists two configurations $\rho$ and $\tau$ such that $\rho \to \tau$ and such that there is a plateau of length at least $p+2$ in $\tau$, and none in $\rho$ ($\mathcal H_2$).
  Let $k$ be the leftmost column of the plateau of length at least $p+2$ in $\tau$, {\em i.e.} for all column $j$ between $k$ and $k+p$, $\tau_j=0$ ($\mathcal H_3$). We will now see that there is no $i$ such that $\rho \overset{i}{\to} \tau$, which completes the proof.
  \begin{itemize}
    \item if $i<k-p$ or $i>k+p+1$ then a firing at $i$ has no influence on columns between $k$ and $k+p+1$ and there is a plateau of length at least $p+2$ in $\rho$, contradicting $\mathcal H_2$.
    \item if $k-p \leq i \leq k$ then according to the rule definition we have $\tau_{i+p}=\rho_{i+p-1}-1$, and from $\mathcal H_3$ $\rho_{i+p}=0$ therefore $\tau_{i+p}<0$ which is not possible (recalled in $\mathcal H_1$).
    \item if $k < i \leq k+p+1$ then according to the rule definition we have $\tau_{i-1}=\rho_{i-1}-p$ and from $\mathcal H_3$ $\rho_{i-1}=0$ therefore $\tau_{i-1}<0$ which again is not possible from $\mathcal H_1$.
  \end{itemize}
\end{proof}

%%%%%%%%%%%%%%%%%%%%%%%%%%%%%%%%%%
%
%	THE SUPPORT OF \PI(N) IS IN \THETA(\SQRT{N})
%
%%%%%%%%%%%%%%%%%%%%%%%%%%%%%%%%%%

\section{The support of $\pi(N)$ is in $\Theta(\sqrt{N})$}\label{a:support}

We give bounds for the maximal index of a non-empty column in the fixed point $\pi(N)$ according to the number $N$ of grains, denoted $w(N)$. The number $w(N)$ can be interpreted as the support or width or size of $\pi(N)$. We consider a general model KSPM($p$) with $p$ a constant integer greater or equal to 1. A formal definition of $w(N)$ is for example $w(N)=w(\pi(N))=\min\{ i | \forall j \geq i, \pi(N)_j=0 \}$. See Figure \ref{fig:frame}.

\begin{lemma}\label{lemma:support}
  The support of $\pi(N)$ is in $\Theta(\sqrt{N})$.
\end{lemma}

\begin{proof}
  The support of $\pi(N)$ is denoted $w(N)$.
  
  Lower bound: $\pi(N)$ is a fixed point, therefore by definition for all index $i$ we have $\pi(N)_i \leq p$. Then,
  $$N \leq \sum \limits_{i=0}^{w(N)} p \cdot i = p \frac{w(N) \cdot (w(N)+1)}{2} < p^2(w(N)+1)^2$$
  hence $\frac{1}{p} \sqrt{N} - 1 < w(N)$.
  
  Upper bound: From Lemma \ref{lemma:plateau}, there is no plateau of length greater than $p+1$. Therefore, for $w(N) \geq p$ we have
  $$N \geq \sum \limits_{i=0}^{\lfloor \frac{w(N)}{p+1} \rfloor} (p+1) \cdot i \geq (p+1)\left(\frac{\left(\frac{w(N)}{p+1} - 1\right)\frac{w(N)}{p+1}}{2}\right) > \left(\frac{w(N)}{p+1} - 1\right)^2$$
  hence $(p+1)\sqrt{N} + p+1 > w(N)$.
\end{proof}

\begin{figure}[!h]
  \centering \includegraphics{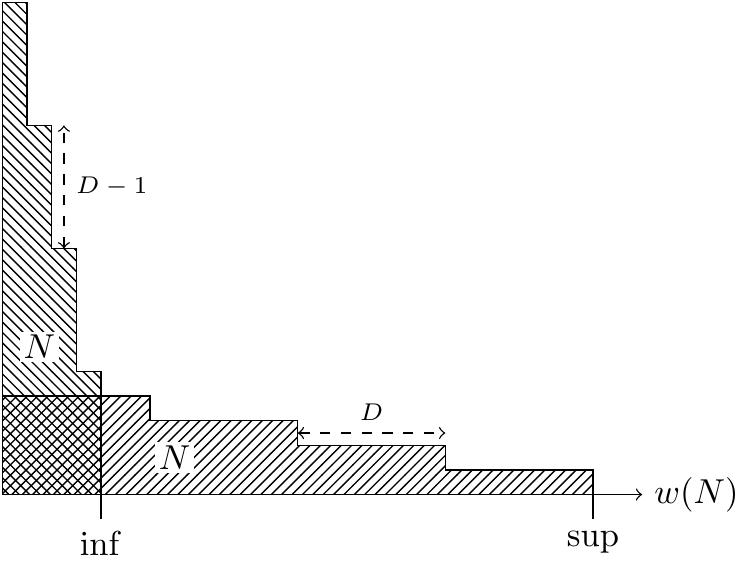}
  \caption{The support of $\pi(N)$ is in $\Theta(\sqrt{N})$. It is lower bounded by the fact that on a stable configuration each slope is at most $p$, and upper bounded by the fact that there is no plateau of length larger than $p+1$.}
  \label{fig:frame}
\end{figure}
   
%%%%%%%%%%%%%%%%%%%%%%%%%%%%%%%%%%
%
%     CONVERGENCE OF $(Y_i)_{i \in \N}$
%
%%%%%%%%%%%%%%%%%%%%%%%%%%%%%%%%%%

%TODO : n -> i

\section{Convergence of $(Y_i)_{i \in \N}$}\label{a:convergence}

We recall that $m_i$ (resp. $\underline{m_i}$, $\overline{m}_i$) is the mean (resp. minimum, maximum) of the components of $Y_i$.

\setcounter{tmp}{\value{lemma}}
\setcounter{lemma}{\value{lemmalarge}}
\begin{lemma}
\setcounter{lemma}{\value{tmp}}
  There exists a constant $\alpha$ and a $n_0$ in $\O(\log N)$ s.t. $\overline{m}_{n_0} - \underline{m}_{n_0} < \alpha$.
\end{lemma}
\begin{proof}
  We start with $Y_0=\,^t(-N,0,\dots,0,a_0)$, thus $\overline{m}_0-\underline{m}_0=N+a_0 \leq \frac{p+1}{p}N$.
  
  The relation linking $Y_i$ to $Y_{i+1}$ is
  $$Y_{i+1} = M \, Y_i + \frac{b_i}{p} \, K$$
  
  Since we want to prove that $Y_i$ converges to a uniform vector close to the mean of its values, we will consider the evolution of the distance to the mean vector associated to $Y_i$, using the uniform vector $M_i=(m_i,\dots,m_i)$ of $\Z^p$. Let $Z_i=Y_i-M_i$, we have
  $$Z_{i+1}=O \, Z_i + \frac{b_i}{p} \, L \text{ ~~~~~where } \begin{array}\{{l}. O=D\,M\\L=D\,K\end{array} \text{ with } D=\begin{pmatrix}1 & & 0\\ & \ddots & \\0 & & 1\end{pmatrix}-\frac{1}{p}\begin{pmatrix} 1 & \dots & 1\\\vdots & \ddots & \vdots\\ 1 & \dots & 1\end{pmatrix}$$
 because $D \, M \, Y_i = D \, M \, D \, Y_i$.
  
  The aim is thus to prove that there exists an $n_0$ in $\O(\log N)$ such that the norm of $Z_{n_0}$ is bounded by a constant.
  
  We express $Z_n$ in terms of $Z_0$ and $(b_i)_{0 \leq i \leq n}$:
  $$Z_{n} = O^n Z_0 + \frac{1}{p} \sum \limits_{i=0}^{n-1} b_i O^{n-1-i} L$$
   
  In order to prove the result, we will see that the linear map $O$ is {\em eventually contracting}, hence it converges exponentially quickly to $\0$, its unique fixed point (\cite{hasselblatt} Corollary 2.6.13). That is, $O^n \, Z_0$ converges to $\0$ exponentially quickly. It then remains to upper bound the norm of the remaining sum by $\alpha$ to get the result.
  
  To prove that $O$ is eventually contracting, it is enough to prove that its {\em spectral radius}\footnote{the maximal absolute value of an eigenvalue of $O$.} is smaller than 1 (\cite{hasselblatt} Corollary 3.3.5). This part is detailed in Appendix \ref{a:eigenZ} using the fact that $M$ is a companion matrix which eigenvalues are upper bounded with a result by Enestr\"om and Kakeya \cite{enestromkakeya}.
    
  Since $\overline{m}_0-\underline{m}_0$ is in $\O(N)$, $\| Z_0 \|_{\infty}$ is also in $\O(N)$ and there exists an $n_0$ in $\O(\log N)$ such that $\| O^{n_0} \, Z_0 \|_{\infty} < 1$.
  
  It remains to upper bound the summation by a constant (we recall that for a matrix $A$, $\| A \|_\infty = \sup \| A \, x\|_\infty$ for $\|x\|_\infty=1$):
  $$\begin{array}{rcl}
  \left\| \frac{1}{p} \sum \limits_{i=0}^{n_0-1} b_i \, O^{n_0-1-i} L \right\|_{\infty} &\leq& \frac{1}{p}\sum \limits_{i=0}^{n_0-1} p\, \| O \|_\infty^{n_0-1-i} \| L\|_\infty\\
  &\leq& \frac{1}{1-\| O \|_\infty} \| L\|_\infty\\
  &\leq& \beta-1
  \end{array}$$
  for some constant $\beta$ independent of $N$. Finally, we have
  $$\| Z_{n_0}\|_\infty \leq \| O^{n_0} Z_0 \|_\infty + \| \frac{1}{p} \sum \limits_{i=0}^{n-1} b_i \, O^{n-1-i} L\|_\infty \leq \beta$$
  and the fact that $\overline{m}_{n_0} - \underline{m}_{n_0} \leq 2 \| Z_{n_0} \|_\infty$ completes the proof with $\alpha=2\beta$.
\end{proof}

%%%%%%%%%%%%%%%%%%%%%%%%%%%%%%%%%%
%
%     DISTINCT ROOTS
%
%%%%%%%%%%%%%%%%%%%%%%%%%%%%%%%%%%

\section{Roots of $R(x)$}\label{a:roots}

Let $R(x)=x^{p-1} + \frac{p-1}{p}\, x^{p-2} + \dots + \frac{2}{p}\, x+ \frac{1}{p}$.

\begin{lemma}\label{lemma:roots}
  $R(x)$ has $p-1$ distinct roots $\lambda_1,\dots,\lambda_{p-1}$ and for all $i$, $\lambda_i \leq \frac{p-1}{p}$.
\end{lemma}

\begin{proof}
  The distinctness of the roots of $S(x)=p\, x^{p-1}R(\frac{1}{x})=x^{p-1}+2\,x^{p-2}+\dots+(p-1)\,x+p$ implies the distinctness of the roots of $R(x)$. The distinctness of the roots of $S(x)$ comes from the fact that $S(x)$ and $S'(x)$ are co-prime. With $a=\frac{-p+1}{p(p+1)}\,x+\frac{1}{p}$ and $b=\frac{1}{p(p+1)}\, x^2-\frac{1}{p(p+1)}\, x$, we get $a\, S(x)+b\, S'(x)=1$. Therefore from Bezout $GCD(R(x),R'(x))=1$, which implies the result.
  
  For the second part of the lemma, a classical result due to Enestr\"om and Kakeya (see for example \cite{enestromkakeya}) concerning the bounds of the moduli of the zeros of polynomials having positive real coefficients states that all the complex roots of $R(x)$ have a moduli smaller or equal to $\frac{p-1}{p}$.
\end{proof}
 
%%%%%%%%%%%%%%%%%%%%%%%%%%%%%%%%%%
%
%     EIGEN VALUES OF DM
%
%%%%%%%%%%%%%%%%%%%%%%%%%%%%%%%%%%

\section{Eigen values of $O=DM$}\label{a:eigenZ}

$M$ is a {\em companion matrix}, its characteristic polynomial is
$$x^p-\sum \limits_{k=0}^{p-1} \frac{1}{p}\, x^k = (x-1)\, R(x)$$
with $R(x)=x^{p-1} + \frac{p-1}{p}\, x^{p-2} + \dots + \frac{2}{p}\, x+ \frac{1}{p}$. From Lemma \ref{lemma:roots} we know that $R(x)$ has $p-1$ distinct roots $\lambda_1,\dots,\lambda_{p-1}$, all comprised between $\frac{1}{p}$ and $\frac{p-1}{p}$. The set of eigenvalues of $M$ is thus $M_\lambda=\{1,\lambda_1,\dots,\lambda_{p-1}\}$. In this section, we prove that the set of eigenvalues of the matrix $O=DM$ is $DM_\lambda = \{0,\lambda_1,\dots,\lambda_{p-1}\}$.

Let $v_0,\dots,v_{p-1}$ be non null eigenvectors respectively associated to the eigenvalues $1,\lambda_1,\dots,\lambda_{p-1}$ of $M$. The case $v_1$ is particular and allows to conclude that $0$ is an eigenvalue of $DM$. The other eigenvectors of $M$ lead to the conclusion that $DM$ also admits the eigenvalues $\lambda_1,\dots,\lambda_{p-1}$.

\begin{itemize}
  \item $DM \, v_0 = D \,v_0$ since the associated eigenvalue is $1$, and $D \, v_0 = \0$ because the eigenspace associated to the eigenvalue 1 is the hyperplan of uniform vectors. As a consequence, 0 is an eigenvalue of $DM$.
  \item For the other eigenvectors, that is, for $1 \leq i \leq p-1$, let $c_i$ be the uniform vector with all its components equal to $\frac{1}{p}\sum_{k=0}^{p-1} v_{i_k}$, with $v_{i_k}$ the $k^{th}$ component of the vector $v_i$. We have $v_i - c_i \neq \0$, and
  $$\begin{array}{rcl}
  DM \, (v_i - c_i) & = & D \, (M \, v_i - M \, c_i)\\
  & = & D \, (\lambda_i \, v_i - c_i)\\
  & = & \lambda_i \, D \, v_i - D \, c_i\\
  & = & \lambda_i (v_i - c_i) - \0
  \end{array}$$
  where the last equality is obtained from the fact that by definition of $D$ we have $D \, v_i = v_i - c_i$. As a consequence, $\lambda_i$ is an eigenvalue of $DM$.
  \end{itemize}
  
  Finally, $DM_\lambda = \{ 0,\lambda_1,\dots,\lambda_p-1 \}$ and the spectral radius of $DM$ is smaller or equal to $\frac{p-1}{p}$.
   
%%%%%%%%%%%%%%%%%%%%%%%%%%%%%%%%%%
%
%     FROM UNIFORM VECTOR Y_N TO WAVE PATTERN
%
%%%%%%%%%%%%%%%%%%%%%%%%%%%%%%%%%%
   
\section{From uniform vector $Y_n$ to wave pattern}\label{a:yn2wave}
   
\setcounter{tmp}{\value{lemma}}
\setcounter{lemma}{\value{lemmayn2wave}}
\begin{lemma}
\setcounter{lemma}{\value{tmp}}
  $Y_n$ is a uniform vector of $\Z^p$ implies
  $$\pi(N)_{[n,\infty[} \in \Big( 0 + ( p \cdot p\!-\!1 \cdot \ldots \cdot 1 ) \Big) ^* 0^\omega$$
\end{lemma}

\begin{proof}
  We will see that the sequence $(b_i)_{i \geq n}$ is determined, or more accurately nearly determined, from this index $n$ for which $Y_n$ is a uniform vector. We will see that if $Y_n$ is a uniform vector, then the value of $b_n$ is 0 or $p$. If it is $0$ then $Y_{n+1}$ is again a uniform vector; if it is $p$, then the sequence $(b_i)_{n \leq i < n+p}$ is determined to be equal to $(p,p-1,\dots,1)$, and $Y_{n+p}$ is once more a uniform vector. Those patterns are thus repeated until the end of the configuration (the $0^\omega$), hence the result.
   
  We concentrate on the sequence of values of $Y_i$. The fact that its components are integers, and especially the last one, will play a crucial role in the determination of the value of $b_i$ because $0 \leq b_i \leq p$ (let us recall that the sequence $b_i$ is the sequence of slopes of the fixed point with $N$ grains, {\em i.e.}, $b_i=\pi(N)_i$).

  We start from the hypothesis that $Y_n=\,^t(\alpha,\dots,\alpha)$, thus from the averaging system's equation (\ref{eq:averaging}) we have $Y_{n+1}=\,^t(\alpha,\dots,\alpha,\alpha+\frac{b_n}{p})$. Since $Y_{n+1}$ is an integer vector and $\alpha$ is an integer, $b_n$ equals 0 or $p$.
  \begin{itemize}
    \item If $b_n=0$ then $Y_{n+1}=\,^t(\alpha,\dots,\alpha)$ and we are back to the same situation, the dilemma goes on: the value of $b_{n+1}$ is not determined, it can be $0$ or $p$.
    \item If $b_n=p$ then $Y_{n+k+1}=\,^t(\alpha,\dots,\alpha,\alpha+1)$ from the relation above. A regular pattern then emerges:
    \begin{itemize} 
      \item if $Y_{n+1}=\,^t(\alpha,\dots,\alpha,\alpha+1)$, then $Y_{n+2}=\,^t(\alpha,\dots,\alpha,\alpha+1,\frac{p\alpha+1+b_{n+1}}{p})$ and it determines $b_{n+1}=p-1$ so that $Y_{n+2}=(\alpha,\dots,\alpha,\alpha+1,\alpha+1)$ is an integer vector;
      \item if $Y_{n+2}=\,^t(\alpha,\dots,\alpha,\alpha+1,\alpha+1)$, then $Y_{n+2}=\,^t(\alpha,\dots,\alpha,\alpha+1,\frac{p\alpha+2+b_{n+2}}{p})$ and it determines $b_{n+2}=p-2$ so that $Y_{n+3}=\,^t(\alpha,\dots,\alpha,\alpha+1,\alpha+1,\alpha+1)$ is an integer vector;
      \item {\em et cetera} we have $b_{n+i}=p-i$ for $0 \leq i < p$, and eventually $Y_{n+p}=\,^t(\alpha+1,\dots,\alpha+1)$ is a uniform vector (note that $Y_0$ has a negative mean, hence $\alpha$ is negative, which is consistent with the $\alpha+1$ we obtain).
    \end{itemize}
  \end{itemize}
  
  Let us conclude with an illustration. The grey node represents a uniform $Y_n$, and arrows are labeled by values of the slope, thus paths starting from the grey node represent possible sequences $(b_i)_{n \leq i}$.
  \begin{center} \includegraphics{fig-automata.pdf} \end{center}
   When $Y_n$ is uniform we are in the grey node, then either:
   \begin{itemize}
     \item $b_n=0$, in which case we are back in a situation where $Y_{n+1}$ is uniform ;
     \item or $b_n=p$, in which case $b_{n+1}=p-1,\dots,b_{n+p-1}=1$ and we are back in a case where $Y_{n+p}$ is a uniform vector.
   \end{itemize}
\end{proof}
   
%%%%%%%%%%%%%%%%%%%%%%%%%%%%%%%%%%
%
%     FIXED POINT EXAMPLE
%
%%%%%%%%%%%%%%%%%%%%%%%%%%%%%%%%%%

\begin{landscape}

%\ \vspace{-2.8cm} %pour que ça tienne sur une page (1/3)
\section{$\pi(2000)$ for $p=4$}\label{a:2000}

Figure \ref{fig:2000} presents some representations of $\pi(2000)$ for $p=4$ used in the developments of the paper: differences of the shot vector (Figure \ref{fig:2000-dsv}), shot vector (Figure \ref{fig:2000-sv}) and height (Figure \ref{fig:2000-h}).
$$\pi(2000)=(4,0,4,1,3,2,4,1,1,3,4,3,4,2,0,1,4,2,2,1,4,3,2,1,0,4,3,2,1,4,3,2,1,4,3,2,1,4,3,2,1,0^\omega)$$
We can notice on Figure \ref{fig:2000-dsv} that the shot vector differences contract towards some ``{\em steps}'' of length $p$, which corresponds to the statement of Lemma \ref{lemma:average} that the vector $Y_n$ becomes uniform exponentially quickly (note that this graphic plots the opposite of the values of the components of $Y_n$). The shot vector representation on Figure \ref{fig:2000-sv} corresponds to the values of the components of $X_n$, which we did not manage to tackle with classical methods. Figure \ref{fig:2000-h} pictures the sandpile configuration on which the wavy shape appears starting from column 20.

\begin{figure}[!h]
%  \vspace{-1cm} %pour que ça tienne sur une page (2/3)
  \centering
  \subfloat[$\pi(2000)$ for $p=4$ represented as a sequence of shot vector differences.]{\label{fig:2000-dsv}\includegraphics[scale=0.75]{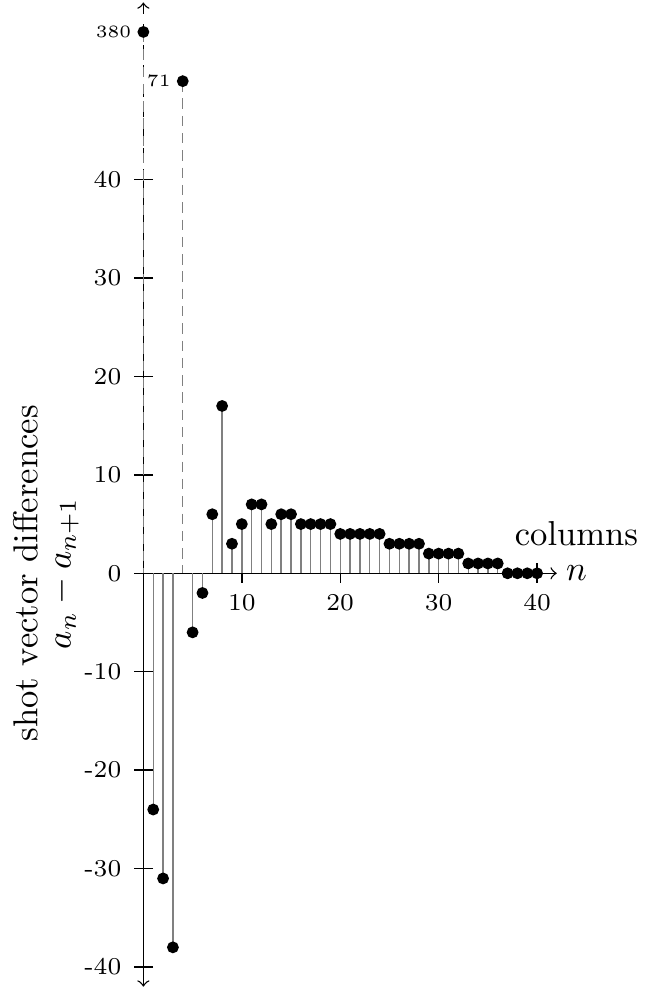}}
  \hspace{1cm}              
  \subfloat[$\pi(2000)$ for $p=4$ represented by its shot vector.]{\label{fig:2000-sv}\includegraphics[scale=0.75]{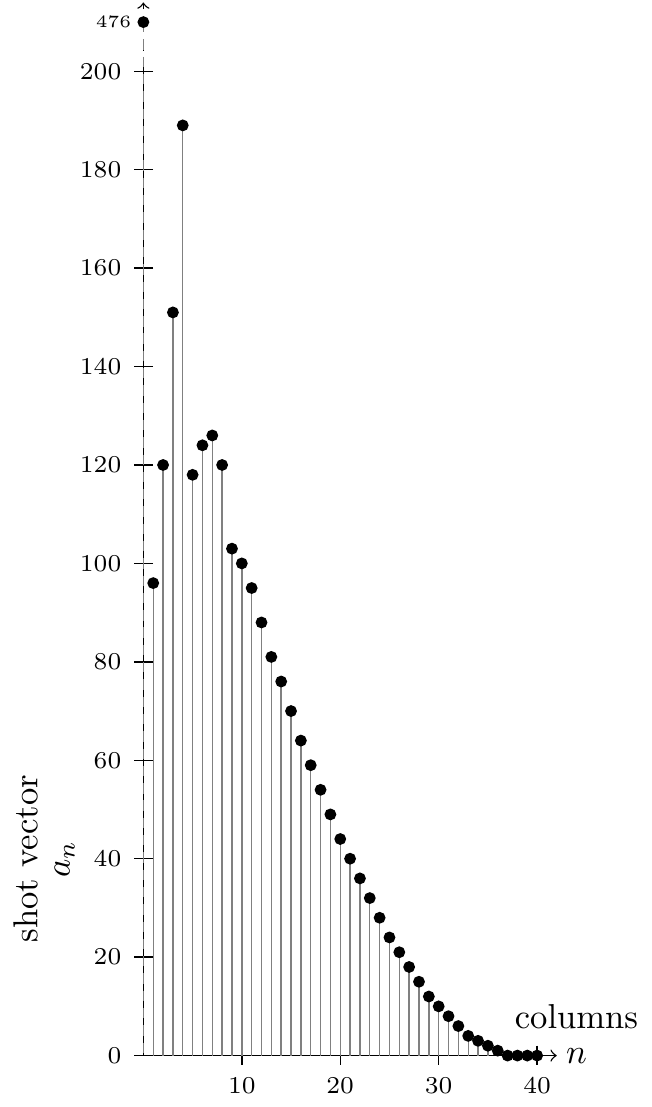}}
  \hspace{1cm}
  \subfloat[$\pi(2000)$ for $p=4$ represented as stacked sand grains.]{\label{fig:2000-h}\includegraphics[scale=0.75]{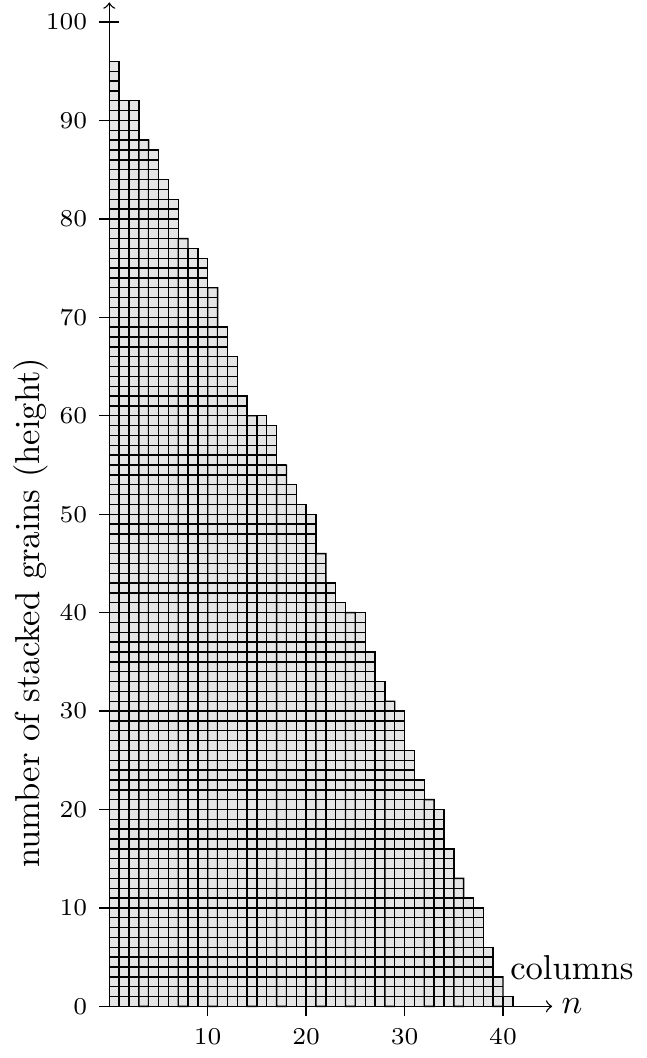}}
  \caption{Representations of $\pi(2000)$ for $p=4$.}
  \label{fig:2000}
  \vspace{-3cm} %pour que ça tienne sur une page (3/3)
\end{figure}

\end{landscape}

%%%%%%%%%%%%%%%%%%%%%%%%%%%%%%%%%%
%
%     GLOBAL DENSITY COLUMN
%
%%%%%%%%%%%%%%%%%%%%%%%%%%%%%%%%%%

\section{Global density column in $\O(\log N)$}\label{a:gdl}

In this section we introduce an important property on avalanches, which, when it is verified starting from an index $n$, leads to regularities in the avalanche process beyond column $n$ (see \cite{LATA} and \cite{MFCS} for details). We will see that a Corollary of Proposition \ref{lemma:average} and Lemma \ref{lemma:yn2wave} is that this property is verified on the $k^{th}$ avalanche\footnote{avalanches are formally defined in Appendix \ref{a:hourglass}.} starting from an index in $\O(\log k)$.

We say that there is a {\em hole} at position $i$ in an avalanche $s^k$ if and only if $i \notin s^k$ and $(i\!+\!1) \in s^k$. An interesting property of an avalanche is the absence of hole from an index $l$, which tells that,
$$\text{there exists an } m \text{ such that }\begin{array}[t]\{{l}. \text{for all } i \text{ with } l \leq i \leq m, \text{ we have } i \in s^k\\ \text{for all } i \text{ with } m < i, \text{ we have } i \notin s^k\end{array}$$
namely, from column $l$, a set of consecutive columns is fired, and nothing else. We say that an avalanche $s^k$ is {\em dense starting from} an index $l$ when $s^k$ contains no hole $i$ with $i \geq l$. We have already explained in \cite{LATA} that this property induces a kind of ``pseudo linearity'' on avalanches, that it somehow ``breaks'' the criticality of avalanche's behavior and let them flow smoothly along the sandpile. Let us introduce a formal definition:

\begin{definition}
  $\mathcal L'(p,k)$ is the minimal column such that the $k^{th}$ avalanche is dense starting at $\mathcal L'(p,k)$:
  $$\mathcal L'(p,k)=\min \{ l \in \N ~|~ \exists m \in \N \text{ such that } \forall l \leq i \leq m, i \in s^k \text{ and } \forall i > m, i \notin s^k \}$$
  Then, the global density column $\mathcal L(p,N)$ is defined as:
  $$\mathcal L(p,N)=\max \{\mathcal L'(p,k) ~|~ k \leq N \}$$ 
\end{definition}

The global density column $\mathcal L(p,N)$ is the smallest column number starting from which the $N$ first avalanches are dense (contain no hole). A Corollary of Proposition \ref{lemma:average} and Lemma \ref{lemma:yn2wave}, conjectured in \cite{LATA} and proven only for $p=2$, is:

\begin{corollary}\label{corollary:snowball}
For all parameter $p$, $\mathcal L(p,N)$ is in $\O(\log N)$.
\end{corollary}
\begin{proof}
  Let $p$ be a fixed parameter. According to Lemma \ref{lemma:plateau} of Appendix \ref{a:plateau} there is no sequence of more than $p+1$ symbols $0$. Consequently, from Proposition \ref{lemma:average} and Lemma \ref{lemma:yn2wave}, there exists $n$ in $\O(\log N)$ and $c \leq p\!+\!1$ such that
  $$\pi(N)_{[n+c,\infty[} \in ( p \cdot p\!-\!1 \cdot \ldots \cdot 1 ) \Big(0 + ( p \cdot p\!-\!1 \cdot \ldots \cdot 1 ) \Big)^* 0^\omega$$
  We will prove that the avalanche $s^{N+1}$ is dense starting from $n+c$, using a result of \cite{LATA} stating that as soon as $p$ consecutive columns has been fired, the avalanche is dense. Since the same reasoning for all $N$, it holds that $\mathcal L'(p,k+1)$ is also in $\O(\log N)$ for all $k<N$ which completes the proof.
  
  Let us consider the following case disjunction, according to the value of the maximal index fired within the $N\!+\!1^{th}$ avalanche, denoted $\max s^{N+1}$.
  \begin{itemize}
    \item If the $N\!+\!1^{th}$ avalanche ends before column $n+c$, formally if $\max s^{N+1} \leq n+c$, then obviously $\mathcal L'(p,N+1) \leq n+c$.
    \item If the $N\!+\!1^{th}$ avalanche ends beyond column $n+c$, formally if $\max s^{N+1} > n+c$, then we consider the dynamic of the avalanche process $s^{N+1}$ on columns greater than $n$, and prove that the set of $p+1$ consecutive columns $n+c$ to $n+c+p$ are all fired. Firstly, from the locality of the rule (it involves columns at distance at most $p$), the stability of $\pi(N)$, and the fact that each column is fired at most once (see \cite{LATA} for a proof), a column can't be fired if none of the $p$ preceding columns has been fired. Secondly, a column $i$ such that $0 < \pi(N)_i < p$ cannot be fired unless its successor $i+1$ has been fired, and conversely will eventually be fired if $i+1$ is. A consequence of the shape of $\pi(N)_{[n+c,\infty[}$ described above is therefore that $n+c$ and $n+c+p$ are fired, and furthermore with $\pi(N)_{n+c+p}=p$ {\em i.e.}, the two first waves are not separated by any symbol $0$ (column $n+c+p$ is fired after it receives one unit of slope when $n+c$ is fired). A simple induction finally shows that columns $n+c+p-1, n+c+p-2,\dots,n+c+1$ are fired, which completes the proof.
  \end{itemize}
\end{proof}

\end{document}